\documentclass[runningheads]{llncs}

\usepackage{amssymb}
\usepackage{graphicx}

\usepackage{url}
\urldef{\mailsc}\path|imai@iec.hiroshima-u.ac.jp,bruno.martin@unice.fr|
\newcommand{\keywords}[1]{\par\addvspace\baselineskip
\noindent\keywordname\enspace\ignorespaces#1}
\begin{document}

\mainmatter  

\title{Simulations between triangular and hexagonal number-conserving
  cellular automata}

\titlerunning{Simulations between triangular and hexagonal NCCAs}

\author{Katsunobu Imai\dag \ddag %
\and  Bruno Martin\ddag}
\authorrunning{K. Imai and B. Martin}

\institute{\dag Graduate School of Engineering, Hiroshima University, Japan.\\
\ddag I3S, CNRS, University of Nice-Sophia Antipolis, France.\\
\mailsc\\
}

\toctitle{Lecture Notes in Computer Science}
\tocauthor{Authors' Instructions}
\maketitle

\begin{abstract}
  A number-conserving cellular automaton is a cellular automaton whose
  states are integers and whose transition function keeps the sum of
  all cells constant throughout its evolution. It can be seen as a
  kind of modelization of the physical conservation laws of mass or
  energy. In this paper, we first propose a necessary condition for
  triangular and hexagonal cellular automata to be
  number-conserving. The local transition function is expressed by the
  sum of arity two functions which can be regarded as 'flows' of
  numbers. The sufficiency is obtained through general results on
  number-conserving cellular automata. Then, using the previous flow
  functions, we can construct effective number-conserving simulations
  between hexagonal cellular automata and triangular cellular
  automata.  \keywords{Cellular automata; Number-conservation.}
\end{abstract}

\section{Introduction}
A number-conserving cellular automaton (NCCA) is a cellular automaton (CA)
such that all states of the cells are represented by integers and the sum
of the numbers (states) of all cells of a global configuration is
preserved throughout the computation. It can be thought as a kind of
model of physical phenomena as, for example, fluid dynamics and
highway traffic flow~\cite{Nagel} and constitutes an alternative to
differential equations.

There is a huge literature published in the domain which witnesses the
great interest in number-conserving cellular automata which gathers
together physicians, computer scientists and mathematicians. Actually,
this particular model of CA applies to phenomena governed by
conservation laws of mass or energy.

Boccara et al.~\cite{Boccara} studied number conservation of
one-dimensional CAs on circular configurations. Durand et
al.~\cite{Durand,Formenti} considered the two-dimensional case and the
relations between several boundary conditions. These results are very
useful for deciding whether a given CA is number-conserving but do not
help much for the design of NCCAs with complex transition rules.

As for the rectangular von Neumann neighborhood
case~\cite{IIIM04,IWNC}, several necessary and sufficient conditions
to be number-conserving are shown.  According to these conditions, the
local function of a rotation-symmetric NCCA is expressed by the sum of
arity two functions as in~\cite{IFIM02}. Designing the functions, we
constructed several NCCAs including a 14-state logically universal
NCCA with rotation-symmetry~\cite{IWNC}, always with square
neighborhoods.

In this paper, we show specific necessary conditions for triangular
and hexagonal CAs to be number-conserving. Under some symmetry
assuptions (rotation symmetry for the triangular case and permutation
symmetry for the hexagonal one), we show that the local transition
function can be decomposed into the sum of several \emph{flow
  functions}, that is, functions only depending upon two
variables. These flow functions are later used to design respective
simulations between hexagonal cellular automata and triangular
cellular automata if we assume both cellular automata to be
permutation-symmetric.

This paper is organized as follows; section~\ref{sec:1} recalls the
classical definitions that will be used; section~\ref{sec:NCCA-tr+hex}
exhibits the necessary conditions for triangular and hexagonal CAs to
be number-conserving. And, finally, in section~\ref{sec:simu}, we
present number-conserving simulations between hexagonal cellular
automata and triangular cellular automata under the permutation
symmetry assumption.

\section{Definitions}
\label{sec:1}
\begin{definition}
  A \emph{deterministic two-dimensional radius one cellular automaton}
  is a 5-tuple defined by $A=(\mathbb{Z}^2, n, Q, f, q)$, where
  $\mathbb{Z}$ is the set of all integers, $n\in\{3,4,6\}$ is the
  number of neighbor cells (which implies a corresponding neighbor
  vector set, which is a finite and ordered set of distinct vectors
  from $\mathbb{Z}^2$:
  $\{\overrightarrow{v_0},\ldots,\overrightarrow{v_n}\}$), $Q$ is a
  non-empty finite set of internal states of each cell, $f : Q^n
  \rightarrow Q$ is a mapping called the \emph{local transition
    function} and $q \in Q$ is a \emph{quiescent state} that satisfies
  $f(q,\cdots,q)=q$.
\end{definition}

A {\it configuration} over $Q$ is a mapping $\alpha : \mathbb{Z}^2
\rightarrow Q$. The set of all configurations over $Q$ is denoted
by Conf($Q$), i.e., ${\rm Conf}(Q) = \{ \alpha | \alpha : \mathbb{Z}^2
\rightarrow Q \}$.  The function $F:{\rm Conf}(Q) \rightarrow {\rm
  Conf}(Q)$ is defined as follows and is called the {\it global
  function} of $A$ induced by $f$:
$\forall\alpha\in{\rm
  Conf}(Q),\forall\overrightarrow{v}\in\mathbb{Z}^2,
F(\alpha)(\overrightarrow{v})=f(\alpha(\overrightarrow{v}+\overrightarrow{v_0}),\ldots,\alpha(\overrightarrow{v}+\overrightarrow{v_n})).$
From now on, we will denote $\overrightarrow{v}$ by $(x,y)$.

Let us denote by $C_F$ the set of \emph{finite configurations}
i.e. which have a finite number of non-quiescent states. A cellular
automaton $A$ is {\it finite number-conserving} (FNC) when it
satisfies
\[\forall \alpha\in C_F,\quad\sum_{(x, y) \in \mathbb{Z}^2} \{ F(\alpha)(x, y) -\alpha(x, y) \} =0.\]
And, according to~\cite{Formenti}, $A$ is FNC if and only if it is
number-conserving.

Next we define some symmetry conditions of common use
(eg. see~\cite{IIIM04}).

\begin{definition}
  CA {\it A} is \emph{rotation-symmetric} if its local function $f$
  satisfies:
\[
\forall g,s_i \in Q \ (1\leq i\leq n), f(g,s_1,\cdots,s_n) = f(g,s_2,\cdots,s_n,s_1),
\]
and {\it A} is \emph{permutation-symmetric} if its local function $f$
satisfies:
\[
\forall g,s_i \in Q \ (i\leq 1\leq n), \forall \pi \in S_n, f(g,s_1,\cdots,s_n) = f(g,s_{\pi(1)},\cdots,s_{\pi(n)}),
\]
where $S_n$ denotes the permutation group with $n$ elements.
\end{definition}

\subsection{Simulation}
\label{sec:simulation}

Below, we propose the definition of a step by step simulation between
two CAs. It expresses that if a CA $A$ simulates each step of CA
$B$ in $\tau$ units of time, there must exist effective applications
between the corresponding configurations~\cite{DAM07}:
\begin{definition}\label{def:simu}
  Let ${\rm Conf}_A$ and ${\rm Conf}_B$ be the two sets of CA
  configurations of (resp.) $A$ and $B$. We say that $A$ simulates
  each step of $B$ in time $\tau$ (and we note $B\stackrel{\tau}{\prec}A$)
 if there exists a constant $\tau\in\mathbb{N}$ and two recursive functions
$\kappa:{\rm Conf}_B\rightarrow{\rm Conf}_A$ and
$\rho:{\rm Conf}_A\rightarrow{\rm Conf}_B$ such that $\kappa\circ\rho
=\mbox{Id}$ and for all $c,c'\in{\rm Conf}_B,$ there exists
$c''\in{\rm Conf}_A$ such that if $c'=F_B(c)$,
$c''=F_{A}^{\tau}(\kappa(c))$ with $\rho(c'')=c'$, where $F_M$
denotes the global transition of CA $M$ and $F_M^t$ the $t$-th
iterate of a global transition of CA $M$.
\end{definition}

Depending upon the value of $\tau$, we say that the simulation is
\emph{elementary} if $\tau=1$ and \emph{simple} if $\tau=O(1)$.

\section{Von Neumann neighborhood number-conserving CA}
\label{sec:NCCA-tr+hex}

Durand et al.~\cite{Durand} proved a general necessary and sufficient
condition for a NCCA with $n\times m$ neighbors to be
number-conserving. With this condition, any local function can be
decomposed into the summation of the local function in which several
arguments are fixed to zero (which is a quiescent state). The drawback
of this general statement is that it does not explicitely represent
neither the movement of the numbers nor symmetries. In the sequel, we
show novel necessary conditions for NCCA in different lattices
structures, namely triangular and hexagonal. The case of the square
grid was already considered in~\cite{IWNC}, where a necessary and
sufficient condition for a von Neumann neighborhood CA to be
number-conserving was shown.

\subsection{Triangular number-conserving cellular automata}
\label{sec:tNCCA}

\begin{figure}[htbp]
\begin{center}
\includegraphics[scale=0.45]{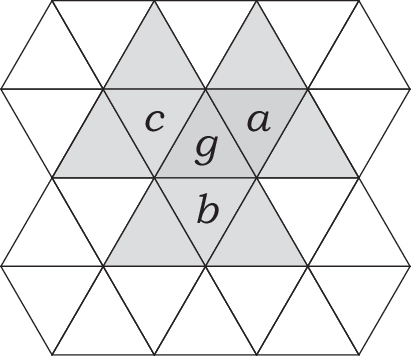}
\caption{A configuration in the triangular case.}
\label{tri-conf}
\end{center}
\end{figure}

\begin{theorem}
  A deterministic two-dimensional rotation-symmetric triangular CA $A
  = (\mathbb{Z}^2, 3, Q, t, q)$ is number-conserving iff $t$ satisfies
\begin{eqnarray*}
&&\exists \varphi :Q^2 \rightarrow \mathbb{Z}, \forall g, a, b, c \in Q, \nonumber\\
&& t(g,a,b,c) = g + \varphi(g, a) + \varphi(g, b) + \varphi(g, c)  \nonumber\\
&& \varphi(g, a) = -\varphi(a, g). \nonumber
\end{eqnarray*}
with $\varphi(g,a)= t(g, a, q, q) - t(g, q, q, q) - t(q, g, q, q) + q$.
\end{theorem}

\begin{proof}
  Let $\delta(g,a,b,c)\equiv
  t(a,g,q,q)+t(b,g,q,q)+t(c,g,q,q)+t(g,a,b,c)+2 t(q,a,q,q)+2
  t(q,b,q,q)+2 t(q,c,q,q)-a-b-c-g-6 q$. With respect to the configuration
  of Fig.\ref{tri-conf}, only shadowed cells change their states in
  the next step. Then for any $g,a,b,c$ in
  $Q$, $\delta(g,a,b,c)=0$ is necessary to preserve number conservation.
  Let's consider the following equation.
\[
  \delta(g,a,b,c)-\delta(g,a,q,q)-\delta(g,q,b,q)-\delta(g,q,q,c)+3 \delta(g,q,q,q)=0
\]
  To satisfy the number-conservation, it is also necessary. Finally, the following
condition is necessary by expanding the equation.
\[
t(g,a,b,c)=g+3 q+t(g,a,q,q)+t(g,q,b,q)+t(g,q,q,c)-3 t(g,q,q,q)-3 t(q,g,q,q)
\]
Let $\varphi(g,a)\equiv t(g, a, q, q) - t(g, q, q, q) - t(q, g, q, q)
+ q$, then
$\varphi(g,a)+\varphi(a,g)= 2q-t(q,g,q,q)-t(q,a,q,q)-t(g,q,q,q)-t(a,q,q,q)+t(g,a,q,q)+t(a,g,q,q) = t(g,a,q,q)-t(g,q,q,q)-t(a,q,q,q)=0$.

We use Durand et al. result~\cite{Formenti} for proving the sufficiency. \qed
\end{proof}

\remark
The condition also holds in the case of permutation-symmetry.

\subsection{Hexagonal number-conserving cellular automata}

\begin{figure}
  \centering
\includegraphics[scale=.45]{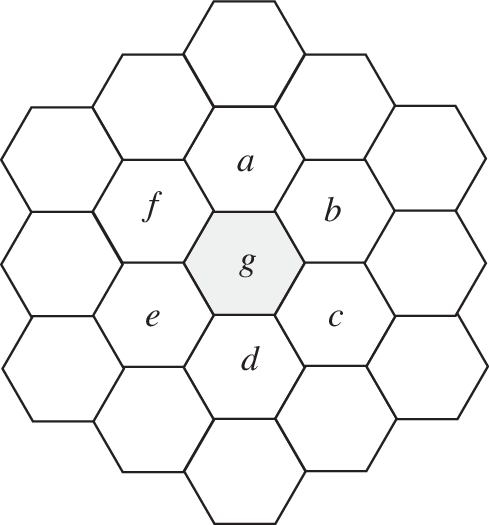}
  \caption{A configuration in the hexagonal case.}
  \label{fig:HexLabel}
\end{figure}
\begin{theorem}\label{th:HPNCCA}
A deterministic two-dimensional permutation-symmetric hexagonal CA,
$\mathcal{H}=(\mathbb{Z}^2,6,Q,\delta,q)$ is number-conserving iff
its local transition function $\delta$ satisfies:
\[
\begin{array}{l}
  \exists\psi:Q^2\rightarrow\mathbb{Z},\quad\forall g,a,b,c,d,e,f,g\in
  Q,\\
  \delta(g, a, b, c, d, e, f)=g+
  \psi(g,a)+\psi(g,b)+\psi(g,c)+\psi(g,d)+\psi(g,e)+\psi(g,f),
\end{array}
\]
with $\psi(g,x)=\delta(g,x,q,q,q,q,q)-\delta(g,q,q,q,q,q,q)-\delta(q,g,q,q,q,q,q)+q$.
\end{theorem}

\begin{proof}
We show that the condition is necessary.
Let us assume that $\mathcal{H}$ is FNC; then according to
Fig.~\ref{fig:HexLabel},
\begin{equation}
\begin{array}{l}
  g+a+b+c+d+f+12q = \delta(g, a, b, c, d, e, f)+\\
  \delta(a, b, f, g, q, q, q)+
  \delta(b, a, c, g, q, q, q)+
  \delta(c, b, d, g, q, q, q)+\\
  \delta(d, c, e, g, q, q, q)+
  \delta(e, d, f, g, q, q, q)+
  \delta(f, a, e, g, q, q, q)+\\
  \delta(q, a, q, q, q, q, q)+
  \delta(q, b, q, q, q, q, q)+
  \delta(q, c, q, q, q, q, q)+\\
  \delta(q, d, q, q, q, q, q)+
  \delta(q, e, q, q, q, q, q)+
  \delta(q, a, q, q, q, q, q)+\\
  \delta(q, a, b, q, q, q, q)+
  \delta(q, b, c, q, q, q, q)+
  \delta(q, c, d, q, q, q, q)+\\
  \delta(q, d, e, q, q, q, q)+
  \delta(q, e, f, q, q, q, q)+
  \delta(q, f, a, q, q, q, q)+
\end{array}
\label{eq:1}
\end{equation}
The local function $\delta(g, a, b, c, d, e, f)$ which satisfies
equation~(\ref{eq:1}) only depends upon terms of four non-quiescent
variables. The idea is to decrease this number of variables down
to two, to finally get the binary $\psi$ function. We first prove
Lemma~\ref{lem:decrease4} which allows to decrease the number of
variables in $\delta$.
\begin{lemma}\label{lem:decrease4}
  In the case of permutation-symmetry, the following equation holds
  for a hexagonal CA, $\mathcal{H}=(\mathbb{Z}^2,6,Q,\delta,q)$:
  $\forall g,x,u,z\in Q,$
\[
    \begin{array}{l}
      \delta(g, x, y, z, q, q, q) = g+x+y+z+12q
      -\delta(x, g, q, q, q, q, q)-
      \delta(q, x, y, g, q, q, q)\\
      -\delta(y, g, q, q, q, q, q)-
      \delta(q, y, z, g, q, q, q)-
      \delta(z, g, q, q, q, q, q)-
      \delta(q, x, z, g, q, q, q)\\
      -3\delta(q,x,q,q,q,q,q)-
      3\delta(q,y,q,q,q,q,q)-
      3\delta(q,z,q,q,q,q,q)
    \end{array}
\]
\end{lemma}
\begin{proof}
Cancelling variables $b,d$ and $f$ by assigning them to the quiescent state $q$
in equation~(\ref{eq:1}) gives, because $\mathcal{H}$ is permutation-symmetric:
\[
\begin{array}{l}
  g+a+c+e+12q=\delta( g,a,c,e,q,q,q ) +
  \delta( a,g,q,q,q,q,q ) +
  \delta( q,a,c,g,q,q,q ) +\\
  \delta( c,g,q,q,q,q,q ) +
  \delta( q,c,e,g,q,q,q ) +
  \delta( e,g,q,q,q,q,q ) +
  \delta( q,a,e,g,q,q,q ) +\\
  3\delta( q,a,q,q,q,q,q ) +
  3\delta( q,c,q,q,q,q,q ) +
  3\delta( q,e,q,q,q,q,q )
\end{array}\]\qed
\end{proof}
\begin{lemma}
  \label{lem:qxy}
  In the case of permutation-symmetry, the following equation holds
  for a hexagonal CA,
  $\mathcal{H}=(\mathbb{Z}^2,6,Q,\delta,q)$: $\forall x,y\in Q$,
  \[
    \begin{array}{ll}
      \delta(q,x,y,q,q,q,q)=&11q+x+y\\
      &-5\delta(q,x,q,q,q,q,q)-5\delta(q,y,q,q,q,q,q)\\
      &-\delta(x,q,q,q,q,q,q)-\delta(y,q,q,q,q,q,q)
    \end{array}
    \]
\end{lemma}

\begin{proof}

Lemma~\ref{lem:qxy} is proved by setting $g,b,c,e$ and $f$ to $q$ in
equation~(\ref{eq:1}).\qed
\end{proof}

\begin{lemma}
    \label{lem:condxy}
  In the case of permutation-symmetry, the following equation holds
  for a hexagonal CA,
  $\mathcal{H}=(\mathbb{Z}^2,6,Q,\delta,q)$: $\forall x,y\in Q$,
  \[
    \begin{array}{ll}
    8q+x+y=&3\delta(q,x,q,q,q,q,q)+3\delta(q,y,q,q,q,q,q)\\
    &+\delta(q,x,y,q,q,q,q)+\delta(q,y,x,q,q,q,q)\\
    &+\delta(x,y,q,q,q,q,q)+\delta(y,x,q,q,q,q,q).
    \end{array}
  \]
\end{lemma}

\begin{proof}
Lemma~\ref{lem:condxy} is proved by replacing $g,c,d$ and $e$ by $q$
in equation~(\ref{eq:1}).\qed
\end{proof}

\begin{lemma}
  \label{lem:condsingle}
  In the case of permutation-symmetry, the following equation holds
  for a hexagonal CA,
  $\mathcal{H}=(\mathbb{Z}^2,6,Q,\delta,q)$: $\forall x,y\in Q$,
  \[
    x=-6q+6\delta(q,x,q,q,q,q,q)+\delta(x,q,q,q,q,q,q).
    \]
\end{lemma}
The proof of Lemma~\ref{lem:condsingle} is straightforward.\\[2ex]

We now prove Theorem~\ref{th:HPNCCA}. We first make a repeated use of
Lemma~\ref{lem:decrease4} by sustracting it from equation~(\ref{eq:1})
with suitable variables substitutions and we obtain
equation~(\ref{eq:eqn1n1}) which only depends upon terms in two
non-quiescent variables.
{\small\begin{equation}
  \label{eq:eqn1n1}
  \begin{array}{l}
    5(a+b+c+d+e+f)+7g+174q+\\
\delta(a,b,q,q,q,q,q)+\delta(a,f,q,q,q,q,q)+\delta(b,a,q,q,q,q,q)+\delta(b,c,q,q,q,q,q)+\\\delta(c,b,q,q,q,q,q)+\delta(c,d,q,q,q,q,q)+\delta(d,c,q,q,q,q,q)+\delta(d,e,q,q,q,q,q)+\\\delta(e,d,q,q,q,q,q)+\delta(e,f,q,q,q,q,q)+\delta(f,a,q,q,q,q,q)+\delta(f,e,q,q,q,q,q)
+\\\delta(g,a,q,q,q,q,q)+\delta(g,b,q,q,q,q,q)+\delta(g,c,q,q,q,q,q)+\delta(g,d,q,q,q,q,q)+\\
\delta(g,e,q,q,q,q,q)+\delta(g,f,q,q,q,q,q)=\delta(g,a,b,c,d,e,f)+
12\delta(g,q,q,q,q,q,q)+\\16(\delta(q,a,q,q,q,q,q)+\delta(q,b,q,q,q,q,q)+\delta(q,c,q,q,q,q,q)+\delta(q,d,q,q,q,q,q)+\\\delta(q,e,q,q,q,q,q)+\delta(q,f,q,q,q,q,q))
+18\delta(q,g,q,q,q,q,q)
+2\delta(q,a,f,q,q,q,q)+\\
2\delta(q,a,g,q,q,q,q)+2\delta(q,b,a,q,q,q,q)+2\delta(q,b,g,q,q,q,q)+2\delta(q,c,b,q,q,q,q)+\\
2\delta(q,c,g,q,q,q,q)+2\delta(q,d,c,q,q,q,q)+2\delta(q,d,g,q,q,q,q)+2\delta(q,e,d,q,q,q,q)+\\
2\delta(q,e,g,q,q,q,q)+2\delta(q,f,e,q,q,q,q)+2\delta(q,f,g,q,q,q,q)+2\delta(q,g,a,q,q,q,q)+\\
2\delta(q,g,b,q,q,q,q)+2\delta(q,g,c,q,q,q,q)+2\delta(q,g,d,q,q,q,q)+2\delta(q,g,e,q,q,q,q)+\\
2\delta(q,g,f,q,q,q,q)
+3(\delta(q,a,b,q,q,q,q)+\delta(q,b,c,q,q,q,q)+\delta(q,c,d,q,q,q,q)+\\\delta(q,d,e,q,q,q,q)+
\delta(q,e,f,q,q,q,q)+\delta(q,f,a,q,q,q,q))
+7(\delta(a,q,q,q,q,q,q)+\\
\delta(b,q,q,q,q,q,q)+\delta(c,q,q,q,q,q,q)+\delta(d,q,q,q,q,q,q)+\delta(e,q,q,q,q,q,q)+\\\delta(f,q,q,q,q,q,q))
+\delta(q,a,c,q,q,q,q)+\delta(q,a,e,q,q,q,q)+\delta(q,b,d,q,q,q,q)+\\\delta(q,b,f,q,q,q,q)+\delta(q,c,e,q,q,q,q)+\delta(q,d,f,q,q,q,q).

  \end{array}
\end{equation}}
We observe that we have two kinds of terms in equation~(\ref{eq:eqn1n1}):
\begin{itemize}
\item $\delta(q,x,y,q,q,q,q)$;
\item $\delta(x,y,q,q,q,q,q)$.
\end{itemize}
The former will be changed by the repeated use of Lemma~\ref{lem:qxy}
and the latter by the repeated use of Lemma~\ref{lem:condxy}. This yields
to equation~(\ref{eq:tmpformula5}).
\begin{equation}
  \label{eq:tmpformula5}
  \begin{array}{l}
    13a+13b+13c+13d+13e+13f+17g+570q+\delta(g,a,b,c,d,e,f)=\\
13\delta(a,q,q,q,q,q,q)+13\delta(b,q,q,q,q,q,q)+13\delta(c,q,q,q,q,q,q)+\\
13\delta(d,q,q,q,q,q,q)+13\delta(e,q,q,q,q,q,q)+13\delta(f,q,q,q,q,q,q)+\\
\delta(g,a,q,q,q,q,q)+\delta(g,b,q,q,q,q,q)+\delta(g,c,q,q,q,q,q)+\\
\delta(g,d,q,q,q,q,q)+\delta(g,e,q,q,q,q,q)+\delta(g,f,q,q,q,q,q)+\\
78\delta(q,c,q,q,q,q,q)+78\delta(q,e,q,q,q,q,q)+
78\delta(q,b,q,q,q,q,q)+\\
78\delta(q,f,q,q,q,q,q)+78\delta(q,a,q,q,q,q,q)+
78\delta(q,d,q,q,q,q,q)+\\
12\delta(g,q,q,q,q,q,q).
  \end{array}
\end{equation}
The sum of the neighbors $13(a+b+c+d+e+f)$ is removed by the repeated
application of Lemma~\ref{lem:condsingle}. Every time we use
Lemma~\ref{lem:condsingle}, we also cancel terms of the form
$78\delta(q,x,q,q,q,q,q)+13\delta(x,q,q,q,q,q,q)$ in the rhs of
equation~\ref{eq:tmpformula5}. All remaining terms in the rhs of
equation~\ref{eq:tmpformula5} are like $\delta(g,x,q,q,q,q,q)$.

We use Durand et al. result~\cite{Formenti} for proving the sufficiency. This proves
Theorem~\ref{th:HPNCCA}.\qed
\end{proof}

All above computations were made using computer algebra systems and
the spreadsheets can be obtained from the authors.

\section{Simulations between hexagonal and triangular NCCAs}
\label{sec:simu}

In this section, we propose effective mutual simulations between
triangular and  hexagonal permutation-symmetric NCCAs.

\subsection{Elementary simulation of a triangular NCCA by a hexagonal NCCA}
\label{sec:tri-hex}

\begin{proposition}
  For any triangular permutation-symmetric NCCA ${\mathcal T}$, there
  is a hexagonal permutation-symmetric NCCA ${\mathcal H}$ such that
  ${\mathcal T} \stackrel{1}{\prec} {\mathcal H}$ and whose transition
  function is surjective.
\end{proposition}
\begin{proof}

  Let $\mathcal{T}=(\mathbb{Z}^2, 3, Q_{{\scriptsize{\mathcal T}}},
  t_{{\scriptsize{\mathcal T}}}, q)$ be a triangular
  permutation-symmetric NCCA with flow function $\varphi$.

  We construct $\mathcal{H}=(\mathbb{Z}^2, 6, Q_{{\scriptsize{\mathcal
        H}}}, t_{{\scriptsize{\mathcal H}}}, q')$ a
  permutation-symmetric hexagonal NCCA by designing its flow function $\psi$.

  Let $q' \notin Q_{{\scriptsize{\mathcal T}}}$ and
  $Q_{{\scriptsize{\mathcal H}}}=Q_{{\scriptsize{\mathcal T}}} \cup
  \{q'\}$; the flow function $\psi$ corresponding to $t_{{\scriptsize{\mathcal
        H}}}$ contains the following values:
\begin{eqnarray*}
  && {\rm For \ each \ } x,y \in Q_{{\scriptsize{\mathcal T}}}, {\rm assign \ } \psi(x,y)=\varphi(x,y), \\
  && {\rm For \ each \ } x \in Q_{{\scriptsize{\mathcal T}}}, {\rm assign \ } ,\psi(x,q')=0. \\
\end{eqnarray*}

Given an initial configuration of ${\mathcal T}$ (see
Fig.~\ref{A3-init-conf}), let the initial configuration of ${\mathcal
  H}$ be as depicted in Fig.~\ref{A6-init-conf}. From the process of
${\mathcal H}$, it is clear that the cells with value $q'$ don't have
any effect and the other non-quiescent cells simulate ${\mathcal
  T}$. Applying the algorithm in~\cite{IIIM04}, it is possible to add
extra-states produced by $t_{{\scriptsize{\mathcal H}}}$ to
$Q_{{\scriptsize{\mathcal H}}}$ and make the local function
$t_{{\scriptsize{\mathcal H}}}$ to be surjective.\qed
\end{proof}

\begin{figure}[htbp]
\begin{center}
\includegraphics[scale=0.45]{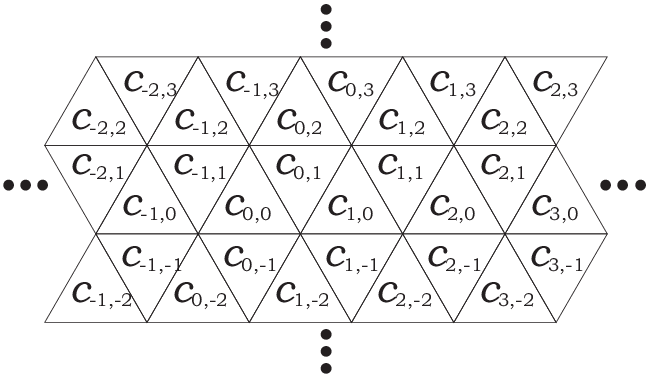}
\caption{An initial configuration of triangular CA ${\mathcal T}$.}
\label{A3-init-conf}
\end{center}
\end{figure}

\begin{figure}[htbp]
\begin{center}
\includegraphics[scale=0.35]{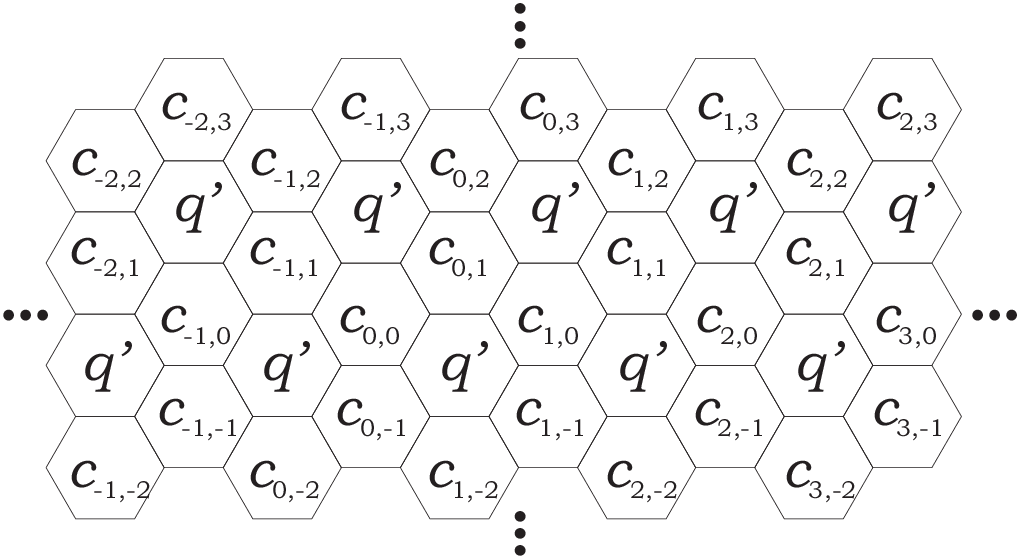}
\caption{An initial configuration of hexagonal CA ${\mathcal H}$.}
\label{A6-init-conf}
\end{center}
\end{figure}

\subsection{Simple simulation of a hexagonal NCCA by a triangular NCCA}
\label{sec:hex-tri}

\begin{proposition}
  For any hexagonal permutation-symmetric NCCA $\mathcal{H}$, there is
  a triangular permutation-symmetric NCCA $\mathcal{T}$ such that
  $\mathcal{H} \stackrel{2}{\prec} \mathcal{T}$ and whose transition
  function is surjective.
\end{proposition}

\begin{proof}
  Let $\mathcal{H}=(\mathbb{Z}^2, 6, Q_{\scriptsize{\mathcal{H}}},
  t_{\scriptsize{\mathcal{H}}}, q)$ be a hexagonal
  permutation-symmetric NCCA.  We assume that
  $Q=\{s_0,s_1,\cdots,s_{m-1}\}$ and its flow function is $\psi$
  and $\mathcal{H}$ has an initial configuration as pictured in
  Fig.~\ref{A6-init-conf}.

\begin{figure}[htbp]
\begin{center}
\includegraphics[scale=0.50]{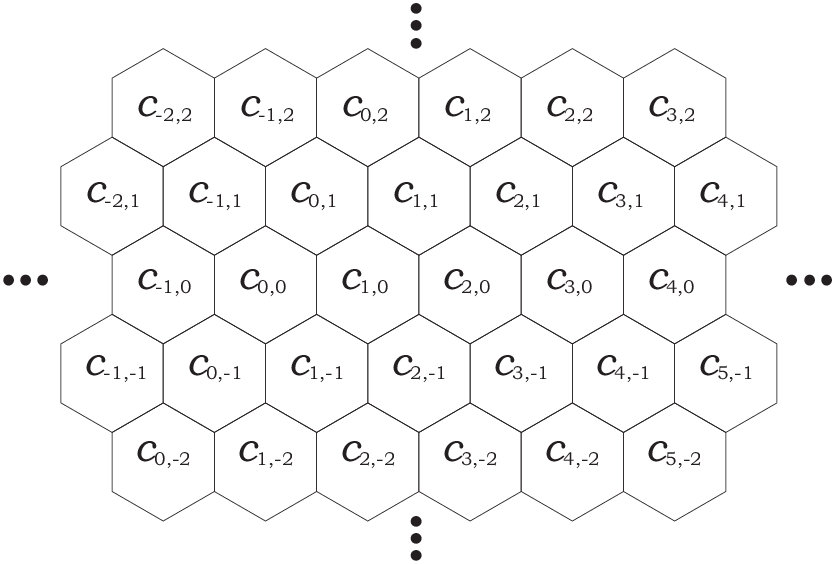}
\caption{An initial configuration of hexagonal CA $\mathcal{H}$.}
\label{A6-init-conf}
\end{center}
\end{figure}

Without loss of generality, we can assume $0<min(s_i)<max(s_j)<M$ for
a constant $M$, because non-positive numbers in
$Q_{\scriptsize{\mathcal{H}}}$ can be erased by adding a constant
value to every state numbers and by changing the arguments of the rules.

We construct $\mathcal{T}=(\mathbb{Z}^2, 3,
Q_{\scriptsize{\mathcal{T}}}, t_{\scriptsize{\mathcal{T}}}, 0)$ a
permutation-symmetric triangular NCCA by designing its flow function
$\varphi$.

First we assign numbers corresponding to each state in
$Q_{\scriptsize{\mathcal{H}}}$.  Let $p_i = 2^{\lceil \lg M \rceil+i}
(i=0,\cdots,m-1)$.

The flow function $\varphi$ of $t_{\scriptsize{\mathcal{T}}}$ contains
the following values:

For each $s_i \in Q_{\scriptsize{\mathcal{H}}}$, assign
\begin{eqnarray*}
&& \varphi(0,s_i)=p_i. \\
&& {\rm For \ each \ combination \ of \ } s_j,s_k \in Q_{\scriptsize{\mathcal{H}}}, \\
&& \varphi(-3p_i,p_i+p_j+p_k)=p_i+\psi(s_i,s_j)+\psi(s_i,s_k). \\
\end{eqnarray*}
All other values of $\varphi$ are $0$. The local function $t_{{\scriptsize{\mathcal T}}}$
can also be extended to be surjective.

The initial configuration of $\mathcal{T}$ is chosen as in Fig.~\ref{P3-init-conf}.
\begin{figure}[htbp]
\begin{center}
\includegraphics[scale=0.50]{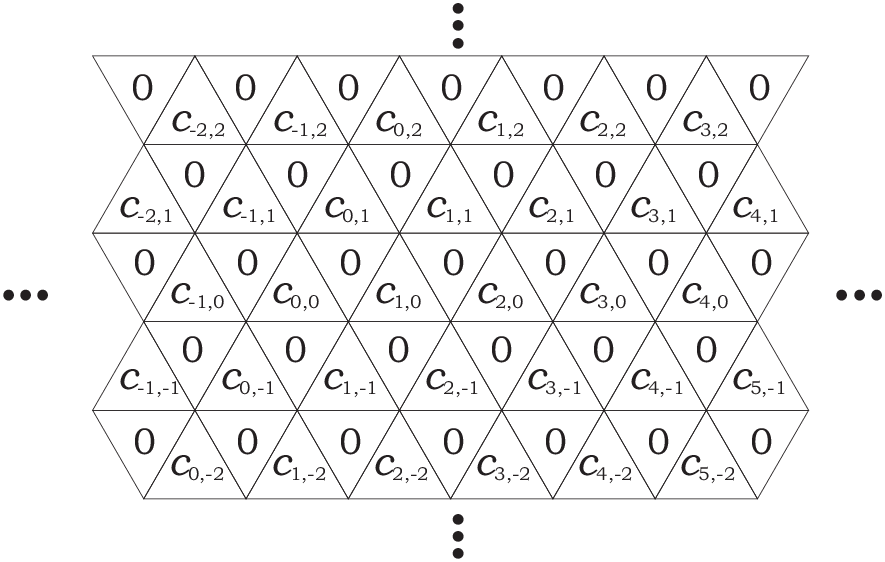}
\caption{An initial configuration of triangular CA $\mathcal{T}$.}
\label{P3-init-conf}
\end{center}
\end{figure}

We briefly explain how the rules work. In the first step, each nonzero
state `$x$' in Fig.~\ref{P3-init-conf} moves three values $p_x$ to the
neighboring zero cells; after this, the zero cells contain the value
$p_x+p_y+p_z$ related to the three neighboring cells
`$x$',`$y$',`$z$'. In the second step, the cell on which we focus (was
`$x$') knows the values of two neighboring cells of $\mathcal{H}$ by the
neighboring cell which value is $p_x+p_y+p_z$ and can move the values of
$\psi(s_x,s_y)+\psi(s_x,s_z)$.\qed
\end{proof}

\section{Conclusion}
In this paper, we have designed flow functions for number-conserving
triangular and hexagonal cellular automata under the
permutation-symmetry condition. This was also generalized to the
rotation-symmetry for triangular NCCA. A simulation between
triangular and hexagonal NCCA, and conversely, were also proposed.

This work can be extended in several ways. First, we'd like to know if
the flow function we proposed in the permutation-symmetry hexagonal
case also holds in the rotation-symmetric case. Then, we aim to go on
in the number-conserving simulation of different neighborhoods like
simulating a von Neumann neighborhood in a square lattice by a
triangular NCCA.

Finally, it might be possible to generalize these results to the case
of number-conserving cellular automata on Cayley graphs and to find
flow functions in several cases and number-conserving simulations as
well.

\subsubsection*{Acknowledgments.}
This work was done while K. Imai was visiting I3S laboratory thanks to
a grant from the french CNRS institution.



\end{document}